\documentclass[12pt, draftclsnofoot, onecolumn]{IEEEtran}
\usepackage{amsmath,cases}
\usepackage{amsthm}
\usepackage{amssymb}
\usepackage{cite}
\usepackage{amsfonts}
\usepackage{enumerate}
\usepackage{url}
\usepackage[final]{graphicx}
\usepackage{comment}

\newtheorem{lem}{Lemma}
\newtheorem{defn}{Definition}

\newtheorem{corollary}{Corollary}

\newtheorem{proposition}{Proposition}

\newcommand{\set}[1]{\mathcal{#1}}
\newcommand{\TBF}{\set{TBF}}
\newcommand{\GGC}{\set{GGC}}
\newcommand{\HCM}{\set{HCM}}
\newcommand{\ID}{\set{ID}}

\newcommand{\DivOrd}[1]{\ensuremath{D_{#1}}}
\newcommand{\snrgap}[1]{\ensuremath{g_{\mathrm{dB}}\left(#1 \right)}}
\newcommand{\cg}[1]{\ensuremath{G_{#1}}}
\newcommand{\codinggain}{array gain}

\newcommand{\orderl}[1]{\leq_{\rm #1}}
\newcommand{\Ind}[1]{I\left[ #1 \right]}

\newcommand{\CDF}[2]{\ensuremath{F_{#1}\left(#2\right)}}
\newcommand{\pdf}[2]{\ensuremath{f_{#1}\left(#2 \right)}}
\newcommand{\MGF}[2]{\ensuremath{\phi_{#1}\left(#2 \right)}}

\newcommand{\re}{\ensuremath{\mathbb{R}}}
\newcommand{\N}{\mathbb{N}}

\newcommand{\AvgPe}[2]{\ensuremath{\overline{P}_{\rm {e}}^{\rm #1}\left(#2 \right)}}

\newcommand{\E}[2]{\ensuremath{\mathbb{E}_{#1}\left[ #2 \right]}}

\newcommand{\D}{ \mathrm{d}}

\newcommand{\g}{\ensuremath{g}}

\newcommand{\cm}{c.m.}

\newcommand{\effch}{channel power}

\begin{document}

\title{A Unified Fading Model Using Infinitely Divisible Distributions}
\author{Adithya Rajan, Cihan Tepedelenlio\u{g}lu, \emph{Senior Member, IEEE}, Ruochen Zeng \thanks{The authors are with the School of Electrical, Computer, and Energy Engineering, Arizona
State University, Tempe, AZ 85287, USA. (Email:
\{arajan2, cihan, zengrc \}@asu.edu). This work was supported in part by the National Science Foundation under Grant CCF 1117041.} } 
\maketitle
\vspace{-2.0cm}
\begin{abstract}
This paper proposes to unify fading distributions by modeling the magnitude-squared of the instantaneous channel gain as an infinitely divisible random variable. A random variable is said to be infinitely divisible, if it can be written as a sum of $n \geq 1$ independent and identically distributed random variables, for each $n$. Infinitely divisible random variables have many interesting mathematical properties, which can be applied in the performance analysis of wireless systems. It is shown that the proposed unification subsumes several unifications of fading distributions previously proposed in the wireless communications literature. In fact, almost every distribution used to model multipath, shadowing and composite multipath/shadowing is shown to be included in the class of infinitely divisible random variables. 
\end{abstract}

\vspace{-8pt}
\section{Introduction}
\label{sec:Intro}
Wireless communications systems are primarily subject to two types of fading effects. Small-scale multipath fading, corresponds to short-term variation of the signal, as it is the result of constructive and destructive combinations of random reflections and scattering of signal components. Models such as Rayleigh, Rician and Nakagami-$m$ distributions have been proposed to capture the distribution of the envelope of the multipath fading random variable (RV). The second type of fading is called shadow fading, which captures the long-term variation of the mean signal level, and is typically modelled as a lognormal RV. In certain scenarios, a composite multipath/shadowing distribution is found to be appropriate, and is modeled as the product of a multipath fading RV and a shadowing RV. It is of interest to see if the typical distributions used for multipath, shadowing, and composite multipath/shadowing may be unified under a common class with desirable analytical properties. Through a unification of fading distributions, it may be possible to obtain canonical expressions for the performance metrics of fading channels, thereby simplifying performance analysis. The unified model may also permit the comparison of two different fading distributions with respect to system performance metrics such as the average symbol error rate (SER), using stochastic orders, which are partial orders on RVs. Therefore, it is compelling to develop a unified study of fading distributions with these goals in mind.

There have been multiple efforts toward unifying fading distributions in the past. Spherically invariant random processes have been proposed for the fading envelope as variance mixture of Rayleigh distributions \cite{yao2004unified}. An unification proposed in \cite{paper:alouini} models the instantaneous {\effch} RV as the product of a gamma and a generalized-gamma RVs. While this model encompasses many known fading distributions, it does not include some distributions such as the Pareto distribution, which has been used in \cite{pun07} to account for interference. Another unification in the literature is the generalized gamma distribution (also known as the Stacy distribution \cite{stacy1962generalization}), which is shown to correspond to the distribution of the envelope, when the received signal is composed of clusters of multipath components, and the receiver possesses power-law nonlinearity \cite{yacoub2007alpha}. Reference \cite{paper:Nstarnakagami} also proposes a fading model, however, it does not include some popularly used fading distributions such as the Rician distribution.
For the first time in the literature, we propose to unify fading distributions by modeling their instantaneous {\effch} as an infinitely divisible (ID) RV. A RV is said to be infinitely divisible, if it can be written as a sum of $n \geq 1$ independent and identically distributed (i.i.d.) RVs, for each $n$. Intuitively, this means that the channel can be viewed as a (time/frequency/antenna) diversity combining system with $n$ i.i.d. channels, for every $n \geq 1$. Infinitely divisible RVs have many interesting mathematical properties \cite{book:steutel2003infinite}, which are relevant in the performance analysis of wireless systems. In this paper, almost every distribution used to model multipath, shadowing and composite multipath/shadowing will be shown to be ID RVs. Most of these distributions are members of a special subclass of ID RVs known as generalized gamma convolutions (GGCs). Consequently, special attention is devoted to this subclass of ID RVs. GGCs are defined as the limits of sums of Gamma RVs, potentially with different parameters, and have been thoroughly studied in \cite{book:bondessonGGC, bondesson1988remarkable, bondesson1990generalized, james2008generalized}. There are numerous applications of GGCs in financial economics \cite{barndorff2001non, carr2001optimal, figueroa2004nonparametric} and risk analysis \cite{burmaster1996introduction, schroder2004risk}. GGCs possess remarkable closure properties under summations or products. These properties make GGCs an attractive model for multipath fading, shadowing or composite models that incorporate both. Despite possessing properties which are useful in fading system analysis, to the best of our knowledge, no applications of GGCs has been found in the wireless communications literature. 

The class of GGCs includes a surprisingly large number of fading models such as the lognormal, Rayleigh, Nakagami-$m$, generalized gamma (Stacy), Pareto, beta, inverse Gaussian and the positive $\alpha$-stable distributions. Moreover, composite multipath/shadowing models such as Rayleigh-lognormal, Nakagami-lognormal, and special cases of the spherically invariant random process, and the generalized model proposed in \cite{paper:alouini} are also GGCs. From a wireless perspective, if a system has GGC distributed instantaneous {\effch}, then the fading channel is equivalent to a coherent linear combination of (possibly infinitely) many independent branches where the instantaneous fading power is gamma distributed. Unifying instantaneous {\effch} distributions as GGCs, it is possible to obtain novel generalized expressions for performance metrics of wireless systems, such as the ergodic capacity and the average SER. Furthermore, this unification reveals new stochastic ordering relations applicable to members of GGC. Lastly, the conditions under which systems such as diversity combining schemes have GGC end-to-end instantaneous {\effch}, are obtained. This facilitates one to obtain bounds on the performance of such systems.

The rest of the paper is organized as follows. In Section \ref{sec:MathPreliminaries} provides the required relevant mathematical background. The system model and the performance metrics under consideration are given in Section \ref{sec:GGCSysmodel}. A short exposition of infinite divisibility highlighting applications in wireless communications is presented in Section \ref{sec:GGC_intro}. Section \ref{sec:GGCandHCMChannels} provides examples of fading distributions which are ID. Performance analysis of GGC fading channels is delegated to Section \ref{sec:GGC_performanceMetrics}. Section \ref{sec:OrderingGGCs} describes the stochastic ordering of GGC RVs. In Section \ref{sec:GGCPerfBoundsSys}, multi-antenna systems with maximum ratio combining (MRC), is considered. Finally, Monte-Carlo simulations are presented in Section \ref{sec:simulations}. 

\subsection{Notation and Conventions}
\label{sec:notations}
The set of real numbers and positive integers are denoted by $\re$ and $\N$ respectively, while all other sets are denoted using script font. For a set $\set{S}$, the indicator function $\Ind{x \in \set{S}} = 1$ if $x \in \set{S}$, and $0$ otherwise. For a Radon measure $\mu(\cdot)$ on $\re^{+}$, $\mu(u)$ is used to represent $\mu([0,u])$. For a RV $X$, the cumulative distribution function (CDF) and the probability density function (PDF) are denoted by $\CDF{X}{\cdot}$ and $\pdf{X}{\cdot}$ respectively. $\E{}{g(X)}$ is used to denote the expectation of the function $g(\cdot)$ over the PDF of $X$. The Laplace transform and Laplace exponent of a nonnegative RV $X$ are defined as $\MGF{X}{\cdot} = \E{}{\exp(-s X)}$ and $\gamma_{X}(s) = -\log \E{}{\exp(-s X)}$ respectively, which are defined for $s \geq 0$. All logarithms are natural logarithms. 
We write $f_{1}(x) \sim f_{2}(x)$, as $x \rightarrow a$ to indicate $\lim_{x \rightarrow a}f_{1}(x)/f_{2}(x) =1$. 
\section{Mathematical Preliminaries}
\label{sec:MathPreliminaries}
\subsection{Bernstein Functions}
A function $g: \re^{+} \rightarrow \re^{+}$ is said to be completely monotone (c.m.) if it has derivatives of all orders, which satisfy $(-1)^{k}\D^{k} g(x)/ \D x^{k} \geq 0$, for $k \in \N \cup  \lbrace 0 \rbrace$. According to Bernstein's theorem, every c.m. function is characterized as a mixture of decaying exponential functions \cite{book:schilling}. Many performance metrics of communication systems such as SERs of $M$-QAM and $M$-PSK are {\cm} functions of the average SNR \cite{paper:adithyaSER}.

If a function $g: \re^{+}\rightarrow \re^{+}$  is such that $\D g(x) / \D x$ is c.m., then it is known as a Bernstein function \cite{book:schilling}. Every Bernstein function can be written in the form
\begin{align}
\label{eqn:bern_characterize}
g(x) = a + bx  + \int\limits_{0}^{\infty} (1-\exp(-sx))\tau(\D s)\;,
\end{align}
for some $a,b \geq 0$, and a non-negative measure $\tau$ on $\re^{+}$ which satisfies $\int_{0}^{1} s \tau(\D s) + \int_{1}^{\infty} \tau( \D s) < \infty$. In this case, $\tau$ is known as the Levy measure of $g$. Bernstein functions are relevant to this work, because the Laplace exponents of non-negative infinitely divisible RVs introduced in Section \ref{sec:GGC_intro}, are Bernstein functions \cite[p. 37]{book:schilling}.

The subset of Bernstein functions, for which $s \tau(s)$ is c.m., constitutes the set of Thorin-Bernstein functions, $\TBF$. Any function in this set admits the following representation, in addition to \eqref{eqn:bern_characterize}:
\begin{align}
\label{eqn:TBF_characterize}
g(x) = a + bx  + \int\limits_{0}^{\infty} \log(1+s x) \mu (\D s)\;,
\end{align}
for some non-negative $a,b$, and a non-negative measure $\mu$ on $\re^{+}$ which satisfies $\int_{0}^{1}|\log s | \mu(\D s) + \int_{1}^{\infty}s \mu(\D s)< \infty$. In this case, $\mu$ is called the Thorin measure of $g$. In this work, we refer to $\int_{0}^{\infty} \mu(\D s)$ as the \emph{mass} of the Thorin measure, or the \emph{Thorin mass}. It will be seen that the Thorin mass of the instantaneous channel power is the diversity order. We will henceforth refer to the Thorin mass as diversity order.

An example of a Thorin-Bernstein function relevant to wireless communications is $g(x) = \log(1+x)$. This is because, if $\rho X$ is viewed as the instantaneous SNR of a channel, where $\rho$ is the average SNR, then $\E{}{g(\rho X)}$ represents the ergodic capacity of the channel.
\subsection{Regular Variation}
Some concepts from regular variation theory will be used in Section \ref{subsec:GGCGenGammaIntro}. Intuitively, regular variation captures asymptotic polynomial-like behavior. 

\begin{defn}

A real valued function $H(x) : \re^{+} \mapsto \re^{+}$ is said to be regularly varying at $\infty$ (at $0$) if $\lim_{x \rightarrow \infty} H(tx)/H(x)$ ($\lim_{x \rightarrow 0} H(tx)/H(x)$) exists, and is equal to $t^{r}$, with $0 < |r| < \infty$. It is said to be slowly varying at $\infty$ (or at $0$) if $r=0$. If $H$ is differentiable with $h(x) = \D H(x)/\D x$, and $h(x) = m x^{m-1} l(x)$, where $l$ is slowly varying at $0$, then $H(x) \sim x^{m} l(x)$, as $x \rightarrow 0$. 
\end{defn}

\subsection{Stochastic Ordering}
\label{sec:StochPreliminaries}
Stochastic orders are binary relations defined on probability distributions which capture intuitive notions like being larger or being more variable. A systematic development of stochastic ordering theory can be found in \cite{shakedbook94}. Some stochastic orders relevant to this paper are now presented. In what follows, we define stochastic orders on measures which might not have unit mass, and therefore are not necessarily probability measures, since this generality will be needed in Section \ref{sec:OrderingGGCs}.

\subsubsection{Convex Order}
\label{subsec:CXorder}
Let $X$ and $Y$ be two random variables such that 
\begin{align}
\E{}{\g(X)} \leq \E{}{\g(Y)}
\end{align}
for all convex functions $g$, provided the expectations exist. Then $X$ is said to be \emph{smaller than} $Y$ \emph{in the convex order} (denoted as $X \orderl{cx} Y$).

\subsubsection{Laplace Transform Order}
\label{subsec:LTorder}
Let $\mu$ be a nonnegative measure. Its Laplace transform is defined as $\int_{0}^{\infty} \exp(-s u)\mu(\D s)$, for $u \geq 0$. This is used to define the LT order between two nonnegative measures. If $\mu_{1}$ and $\mu_{2}$ are two nonnegative measures, then $\mu_{1}$ is said to be dominated by $\mu_{2}$ in the Laplace transform sense (denoted by $\mu_{1} \orderl{Lt} \mu_{2}$) whenever:
\begin{align}
\label{eqn:LT_expectations}
\int_{0}^{\infty} \exp(-\rho s) \mu_{1}(\D s) \geq \int_{0}^{\infty} \exp(-\rho s) \mu_{2}(\D s), \; \forall \rho \geq 0 \;.
\end{align}

This definition of the LT order is somewhat more general than definition found in \cite[Chapter 5]{shakedbook94}, which applies to only probability measures. When $\mu_{1}$ and $\mu_{2}$ are probability measures corresponding to RVs $X$ and $Y$ respectively, useful properties are obtained. $\E{}{\g(X)} \geq \E{}{\g(Y)}$ for all {\cm} functions $g(\cdot)$ is equivalent to $X \orderl{Lt} Y $ \cite[pp. 96]{shakedbook94}. A similar result with a reversal in the inequality states that $\E{}{\g(X)} \leq \E{}{\g(Y)}$ for all Bernstein functions $\g(\cdot)$ is equivalent to $X \orderl{Lt} Y$. 



\subsubsection{Shannon Transform Order}
In what follows, a generalized version of the ergodic capacity order \cite{paper:adithyaCapacityconf} applicable to order nonnegative measures is described. 

Let $\mu$ be a non-negative measure on $\re^{+}$. Its Shannon transform is defined as $\int_{0}^{\infty} \log (1+\rho/s) \mu (\D s)$. If $\mu_{1}$ and $\mu_{2}$ are two non-negative measures for which the Shannon transforms exist $\forall \rho \geq 0$, and are finite, then we say $\mu_{1}$ is dominated by $\mu_{2}$ in the Shannon transform sense, and write $\mu_{1} \orderl{S} \mu_{2}$ to denote $\int_{0}^{\infty} \log(1+\rho/s) \mu_{1} (\D s)  \leq \int_{0}^{\infty} \log (1+\rho/s) \mu_{2} (\D s)$, for all $\rho \geq 0$. 

Some properties of the Shannon transform order are the following:\\
\textbf{P1}: If $\int_{0}^{\infty}\mu_{1}(\D s) = \int_{0}^{\infty}\mu_{2}(\D s) < \infty$, then $\mu_{1} \orderl{Lt} \mu_{2} \implies \mu_{1} \orderl{S} \mu_{2}$\\
\textbf{P2}: $\mu_{1} \orderl{S} \mu_{2} \iff \int_{0}^{\infty} \sigma (s) \mu_{1}(\D s)  \leq \int_{0}^{\infty} \sigma (s) \mu_{2}(\D s)  $, for all $\lbrace \sigma \in \TBF | \sigma(0) = 0 \rbrace$.\\
\textbf{P3}: Let $\int_{0}^{\infty} \mu_{1}(\D s) \leq \int_{0}^{\infty} \mu_{2}(\D s)$. Then $\mu_{1} \orderl{S} \mu_{2} \iff \int_{0}^{\infty}  \sigma (s) \mu_{1}(\D s) \leq \int_{0}^{\infty} \sigma (s) \mu_{2}(\D s) $, for all $\sigma \in \TBF$.
\vspace{-8pt}
\section{System Model}
\label{sec:GGCSysmodel}
The system model under consideration is the following:
\begin{align}
\label{eqn:GGCsystemModel}
y = \sqrt{\rho X}b + w\;,
\end{align}
where $y$ is the received signal, $\rho \geq 0$ represents the average signal to noise ratio (SNR), $b$ is the input symbol to the fading channel, $w$ is circularly symmetric additive white Gaussian noise with zero mean and unit variance. In \eqref{eqn:GGCsystemModel}, $X$ represents the magnitude square of the complex baseband equivalent fading coefficient, and is called as the instantaneous {\effch}. We normalize $\E{}{X}=1$. Unifying the various distributions of $X$ under the assumptions of multipath fading, shadow fading and composite multipath/shadowing is the focus of this work. The unification permits generalized performance analysis of the system with respect to metric such as the average SER, which is defined in Section \ref{sec:GGCasympSER}. 

\section{Infinitely Divisible Fading Distributions}
\label{sec:GGC_intro}
\subsection{Infinite Divisibility}
A probability distribution is said to be infinitely divisible (ID) if, for each $n \geq 1$, it can be decomposed into $n$ identical convolution factors. Some examples of such distributions are the Gamma, the Poisson and the Gaussian distributions. In this paper, we refer to RVs with ID distributions as ID RVs. Non-negative ID RVs are characterized by the fact that the negative logarithm of the real-valued Laplace transform, i.e., the Laplace exponents of such RVs are Bernstein functions. The set of all nonnegative ID RVs is denoted by $\ID$. In a wireless communications context, the resulting effect of coherent combination of contributions from multiple different sources can be viewed as fading nonnegative ID RVs. However, even shadow fading distributions like the log-normal which arise from infinite products of RVs turn out to be infinitely divisible. Indeed, almost all known instantaneous {\effch} distributions are ID, with examples listed in Section \ref{sec:GGCandHCMChannels}. The infinite divisibility of the lognormal distribution, proved by Thorin \cite{thorin1977infinite} revealed a very rich class of ID probability distributions known as generalized gamma convolutions (GGC), which possess extremely useful mathematical properties applicable to wireless systems. In addition, several popularly used fading distributions are not only ID, but also GGC. This fact motivates us to consider this class of ID distributions, which is discussed next. 
\subsection{Generalized Gamma Convolutions}
\label{subsec:GGCGenGammaIntro}
The subset of non-negative ID distributions for which the Laplace exponent belongs to $\TBF$, constitutes the set $\GGC$ of GGC probability distributions. In other words, the LT of the RV admits the form 
\begin{align}
\label{eqn:MGF_GGC_gen}
\phi(s) = \exp \left(-as-\int\limits_{0}^{\infty}\log(1+s/u) \mu(\D u) \right)\;,
\end{align} 
where the parameter $a \geq 0$, and the nonnegative measure $\mu(\cdot)$ uniquely characterize the GGC distribution. In this paper, we call the Thorin measure of the Laplace exponent of the GGC (e.g. $\mu$ in \eqref{eqn:MGF_GGC_gen}) as the Thorin measure of the GGC. The degenerate GGC distribution is obtained with $\mu(\D u) = 0$. 

GGCs are closely linked with Gamma RVs having density
\begin{align}
\pdf{X}{x} = \theta^{D}x^{D-1}\exp(-\theta x)/\Gamma(D),
\end{align}
for $x \in \re^{+}$, where $\theta>0$ is known as the rate parameter, and $D>0$ is known as the shape parameter, denoted as $Gamma(D,\theta)$. In this context, a arbitrary GGC is a weak limit of a sequence of convolutions of gamma distributions with suitably chosen parameters. In wireless communications, the Thorin measure of a GGC has a physical interpretation: Consider the case where a continuum of fading channels with Gamma distributed instantaneous power is coherently summed up. If the sequence of shape parameters corresponding to each rate parameter $u$ in $[0,\infty)$ of the Gamma distributed channels is $\mu(u)$, then the Thorin measure of the GGC channel is $\mu$. 

In what follows, we will restrict our attention to nondegenerate GGCs with $a=0$, because $a$ merely represents a translation of the instantaneous {\effch} PDF, and performance analysis for such channels can be obtained from systems where $a=0$.

The variation exponent of the LT of the GGC near infinity will correspond to the diversity order, and this quantity relates to the Thorin measure as follows:
\begin{lem}{{\cite[Theorem 4.1.4]{book:bondessonGGC}}}
\label{lem:GGC_thorinmass}
The LT of a non-degenerate GGC is regularly varying near $\infty$, and the variation exponent equals the diversity order. That is, $\lim_{\rho \rightarrow \infty} \MGF{}{t \rho}/\MGF{}{\rho} = t^{-D }$, where
$
D = \int\limits_{0}^{\infty}\mu (\D s) 
$.
\end{lem}
The variation exponent of the LT of a GGC near $\infty$ will have applications in asymptotic performance analysis, as discussed in Section \ref{sec:GGC_performanceMetrics}. The Thorin measure also determines the expected value of a GGC RV:
\begin{proposition}
\label{thm:AvgPwr}
Let $X$ be a GGC. Then $\E{}{X} = \int_{0}^{\infty} 1/u \mu(\D u) $, if the integral is finite.
\end{proposition}
\begin{IEEEproof}
It is easy to show for any nonnegative RV that $\E{}{X} = -\lim_{\rho\rightarrow 0} \D \MGF{X}{\rho}/\D \rho$. Since the LT of a GGC is of the form $\MGF{X}{\rho} = \exp(-\gamma(\rho))$, $\E{}{X} = \lim_{\rho \rightarrow 0} \D \gamma(\rho)/\D \rho$, where $\gamma(\cdot)$ is the Laplace exponent of the GGC. Representing $\gamma$ in terms of the Thorin measure, 
\begin{align}
\E{}{X} &= \lim_{\rho \rightarrow 0} \frac{\D \int\limits_{0}^{\infty} \log(1+s/u) \mu(\D u)}{\D \rho} \nonumber \\
& = \lim_{\rho \rightarrow 0} \int\limits_{0}^{\infty} \frac{1/u}{1+\rho/u} \mu(\D u)\;,
\end{align}
where $\mu$ is the Thorin measure. The limit and integral can be interchanged whenever 
$
\int_{0}^{\infty} u^{-1} \mu(\D u) < \infty 
$, because under this condition dominated convergence is satisfied. This completes the proof.
\end{IEEEproof} 
If $X$ represents the instantaneous {\effch} in a fading system, then Proposition \ref{thm:AvgPwr} gives the average channel power of the system. As mentioned in Section \ref{sec:GGCSysmodel}, we assume that the mean of the GGC exists and is normalized to $1$, as per the discussion after \eqref{eqn:GGCsystemModel}. In doing so, we ensure that $\rho$ in \eqref{eqn:GGCsystemModel} reflects the average signal to noise ratio of the system.

It is straightforward to see that, if $X_{1}$ and $X_{2}$ are two GGCs with Thorin measures $\mu_{1}$ and $\mu_{2}$, then $X_{1}+X_{2}$ is a GGC with a Thorin measure $\mu_{1}+\mu_{2}$.  

Some useful properties of non-degenerate GGCs are now presented.

The density of a GGC always exists \cite[Theorem 4.1.2, Theorem 4.1.3]{book:bondessonGGC}, and can be written in terms of {\cm} functions as follows:
\begin{lem}
\label{prop:GGC_density}{{\cite[Theorem 4.1.1, Theorem 4.1.2]{book:bondessonGGC}}}
The density of a GGC with diversity order $D < \infty$ can be written in the form 
\begin{align}
\label{eqn:GGCPDF_canonical}
\pdf{X}{x} = x^{D -1}h(x) \;,
\end{align}
where $h(\cdot)$ is a c.m. function, which is slowly varying at $0$. Moreover, $h(0+)$ is finite if and only if $\int_{1}^{\infty} \log u \mu(\D u) < \infty$, and in this case 
\begin{align}
\label{eqn:hofzero}
h(0+) = \frac{1}{\Gamma(D)} \exp\left( \int\limits_{0}^{\infty} \log u \mu(\D u) \right) \;.
\end{align}
\end{lem}
The representation of the PDF of a GGC in the form \eqref{eqn:GGCPDF_canonical} is particularly useful to obtain asymptotic performance metrics as shown in Section \ref{sec:GGCasympSER}. GGCs are also mixtures of Gamma distributions, as summarized in the following Lemma:
\begin{lem}{{\cite[Theorem 4.1.1]{book:bondessonGGC}}}
\label{lem:GGC_gammaMixture}
If $X$ is a GGC with diversity order $D < \infty$, then $X = A Z$, where $Z \sim Gamma(D,1)$ and $A \geq 0$ are independent. Moreover, $h(x)$ in \eqref{eqn:GGCPDF_canonical} is given by 
\begin{align}
\label{eqn:GGCGammahxandA}
h(x) = (\Gamma(D))^{-1}\E{}{A^{-D}\exp(-x/A)}\;.
\end{align}
\end{lem}
Interpreting $X$ as the instantaneous {\effch} of a GGC channel, from Lemma \ref{lem:GGC_gammaMixture} it is straightforward to see that every GGC channel with finite diversity order is equivalent to a variance mixture of a Nakagami-$m$ channel. This is because the instantaneous {\effch} corresponding to Nakagami-$m$ fading is Gamma distributed \cite{tega04}. It is shown in Section \ref{sec:GGCandHCMChannels} that several fading distributions such as Nakagami-$m$, Pareto and generalized gamma distributions have instantaneous {\effch}s which are GGC RVs, and the corresponding diversity orders are finite.     
 
\subsection{Hyperbolically Completely Monotone RVs}
A subset of $\GGC$, consisting of the distributions for which the density is a hyperbolically completely monotone (HCM) function is denoted by $\HCM$ ($\HCM \subset \GGC \subset \ID$), and the RVs corresponding to distributions in this set are called HCM RVs. A function $g$ is defined to be a HCM function if $g(uv)g(u/v)$ can be written as \cite[p. 68]{book:bondessonGGC}
\begin{align}
g(uv)g(u/v) = \int\limits_{0}^{\infty}\exp(-\lambda u (v+v^{-1}))K(\D \lambda ; u)\;, K(\D \lambda ; u) \geq 0\;,
\end{align}
for each $u>0$. HCM RVs enjoy the following properties under products and quotients of RVs, which will be useful in obtaining closure results for wireless systems involving multiple fading RVs:
\begin{lem}
\label{lem:props_GGC_HCM}
If $X_{i},i=1,2,\dots$ are independent HCM RVs, and $Y$ is a GGC independent of $X_{i}$, then
\begin{enumerate}[(i)]
\item $X_{1}^{q} \in  \HCM$, for $|q| \geq 1$.
\item $X_{1}X_{2} \in \HCM$.
\item $X_{1} Y \in \GGC$.
\item $\sum_{j}\prod_{i}X_{i,j}^{q_{i,j}} \in \GGC$, for $|q_{i,j}| \geq 1$.
\end{enumerate}
\end{lem}

\subsection{Examples of ID Fading Distributions}
\label{sec:GGCandHCMChannels}
In this work, we view instantaneous {\effch} distributions as members of $\ID$. To justify that the $\ID$ class is indeed a unification of fading distributions, many of the popularly used fading channels are shown to have ID instantaneous {\effch}s. Among these channels, distributions for which the instantaneous {\effch} is a GGC or HCM RV are identified, as this fact will be useful in the analysis of systems, where the overall instantaneous {\effch} RV involves a sum or product of two or more RVs (See Section \ref{sec:GGCPerfBoundsSys}). Explicit examples of RVs which are ID, GGC and HCM are now listed.
\subsubsection{Rician-$K$ Fading}
\label{sec:GGCRicianexample}
Rician-$K$ distribution is used to model the fading envelope in certain line-of-sight scenarios \cite[p. 21]{simon_alouini05}. For this case, if $X$ is a nonnegative RV corresponding to the instantaneous {\effch}, then $\sqrt{X}$ is Rice distributed with parameter $K$. The LT of $X$ is given by \cite[p. 19]{simon_alouini05}
\begin{align}
\label{eqn:GGCRicianIDstep1}
\MGF{X}{s} = \frac{1+K}{1+K+s} \exp\left( - \frac{K s}{1+K+s}\right)\;.
\end{align}
It is now shown that $X$ is an ID RV. To begin with, the Laplace exponent of $X$ is obtained from \eqref{eqn:GGCRicianIDstep1} as 
\begin{align}
\gamma_{X}(s) = \frac{K s}{1+K+s}+ \log\left(1+ \frac{s}{1+K} \right)\;.
\end{align}
Using Frullani's representation $\log(1+x) = \int_{0}^{\infty}(1-\exp(-x u)) \exp(-u)/u \D u$, for $ x \geq 0$ \cite[p. 6]{book:lebedev72}, it is straightforward to show that $\gamma_{X}(s)$ can be written in the form \eqref{eqn:bern_characterize} with $a=b=0$ and 
\begin{align}
\label{eqn:GGCRicianIDTau}
\tau(s) = \exp(-(1+K)s)\left( K(1+K)+\frac{1}{s}\right);.
\end{align}
Thus, $\gamma_{X}(s)$ is a Bernstein function, and consequently, $X$ is an ID RV.
However, this is not a GGC because the Laplace exponent $\gamma_{X}(\cdot)$ is not a Thorin-Bernstein function. To see this, recall from Section \ref{sec:MathPreliminaries} that a Bernstein function with Levy measure $\tau$ is a Thorin-Bernstein function, if and only if $s \tau(s)$ is a {\cm} function. However, by differentiating $s \tau(s)$, where $\tau(s)$ is defined in \eqref{eqn:GGCRicianIDTau}, it is seen that the first derivative is not non-positive at all values of $s$, which shows that $s \tau(s)$ is not {\cm}.

The instantaneous {\effch} corresponding to Rician-$K$ fading is therefore an ID RV, but not a GGC. 
\subsubsection{Nakagami-$m$ Fading}
\label{sec:GammaDistGGC}
Nakagami-$m$ fading envelop, which corresponds to a Gamma distributed instantaneous SNR, is used to model line of sight channels. In addition, the Gamma distribution has also been used to model long-term fading effects in the literature \cite{alouini1999unified}. Moreover, the exponential distribution, which is a special case of a Gamma RV with unit shape and rate parameters, arises as the instantaneous fading power distribution in Rayleigh fading channels. 

The instantaneous {\effch} $X$ of Nakagami-$m$ fading is Gamma distributed with shape $m$ and rate $m$\cite{tega04}. In this case, $X$ is a HCM, GGC and ID RV. To see that $X$ is a GGC, observe that the LT of $X$ is given by $\MGF{X}{s} = \exp(-m \log(1+s/m))$, which satisfies the representation of the LT of a GGC specified in \eqref{eqn:MGF_GGC_gen}. It is then straightforward to show using \eqref{eqn:MGF_GGC_gen} that the density of the Thorin measure in this case, is given by
\begin{align}
\label{eqn:gammaThorinMeas}
\D\mu(u) = m \Ind{\lbrace m \rbrace}\;,
\end{align}
and the Thorin mass is $m$, as can be verified by rewriting the PDF in the form \eqref{eqn:GGCPDF_canonical}. Further, $h(\cdot)$ is given by $h(x) = m \exp(-m x)/\Gamma(m)$. The proof that $X$ is a HCM RV is given in \cite[p. 75]{book:bondessonGGC}, and consequently a GGC and ID RV. 

\subsubsection{Nakagami-$q$ (Hoyt) Fading}
Nakagami-$q$ (Hoyt) fading is observed in scenarios such as satellite communications subject to ionospheric scintillation \cite{chytil1967distribution}. Special cases of Hoyt fading are Rayleigh fading ($q=1$) and one-sided Gaussian ($q=0$). 

The in-phase and quadrature components of the complex baseband fading RV for the case of Nakagami-$q$ fading are independent zero mean Gaussian RVs with unequal variances, such that the ratio of the quadrature variance to the in-phase variance equals $q$ \cite{hoyt1947probability, hajri2009study}. Further, the average channel power is unity if the sum of the variances is $1$. As a result, the instantaneous {\effch} $X$ can be written as
\begin{align}
\label{eqn:GGCHoytFadingX}
X = \frac{1}{1+q^{2}}W_{1}^{2}+ \frac{q^{2}}{1+q^{2}}W_{2}^{2} \;,
\end{align}
where $W_{1} \sim \mathcal{N}(0,1)$ and $W_{2} \mathcal{N}(0,1)$ are independent, and $0 \leq q \leq 1$. 

It is now argued that $X$ is a GGC and hence also ID. Toward this end, first observe that $X_{i}$, $i=1,2$ is a GGC, because its Laplace exponent is given by $\gamma(s) = (1/2)\log(1+2s)$, which is a Thorin-Bernstein function. It then follows from \eqref{eqn:GGCHoytFadingX} that $X$ is a GGC, because GGCs are closed under nonnegative scaling of RVs and addition of RVs. Moreover, using the identity $\MGF{aX}{s} = \MGF{X}{as}$ for any nonnegative RV $X$, it can be shown that the density of the Thorin measure is given by
\begin{align}
\label{eqn:GGCHoytThorinMEas}
\frac{\D\mu(u)}{\D u} = \frac{1}{2(1+q^{2})}\Ind{ \lbrace \frac{1+q^{2}}{2} \rbrace} + \frac{q^{2}}{2(1+q^{2})}\Ind{ \lbrace  \frac{1+q^{2}}{2q^{2}} \rbrace}\;,
\end{align}
and the Thorin mass is $1$. It is also noted that $X$ is not HCM, because the Thorin measure of a HCM RV is allowed to have at most one atom \cite[p. 88]{book:bondessonGGC}, which is not the case for $X$, as observed from \eqref{eqn:GGCHoytThorinMEas}. Therefore, $X$ for Hoyt fading is a GGC and ID RV, but not a HCM RV.

\subsubsection{General Lognormal Distribution}
The lognormal distribution is commonly used to model the long-term shadowing effect in wireless channels \cite{simon_alouini05}.

The density of the general lognormal RV is given by
\begin{align}
\label{eqn:logNormalPDF}
\pdf{X}{x} = c x^{\mu/\sigma^{2} -1}\exp\left( -\frac{(\log x)^{2}}{2 \sigma^{2}}\right)\;,
\end{align}
with $\mu \in \re, \sigma >0$, and $c>0$ is a normalizing constant. It has been shown in \cite[p.~74]{book:bondessonGGC} that the general lognormal distribution belongs to $\HCM$, and hence also $\GGC$ and $\ID$. Although the Thorin measure is unknown, and the diversity order can be calculated to be $\infty$ \cite[p.~74]{book:bondessonGGC}. 

In wireless systems, when multiple channels with the lognormal shadowing and $\HCM$ fading effects combining together on each individual channel, the total received channel gain using MRC can be shown to be $GGC$ RV using Lemma \ref{lem:props_GGC_HCM}.  
\subsubsection{Generalized Gamma (Stacy) Distribution}
The instantaneous {\effch} $X$ is generalized Gamma distributed, when the envelope of the fading amplitude is modelled as a generalized Nakagami-$m$ RV \cite{alouini1999unified}. The Weibull distribution is a special case of the generalized Gamma distribution, and has been used to approximate the multipath wireless channel from channel measurements \cite{sagias2004channel}. Other distributions which are special cases of the generalized Gamma distribution are the inverse-Gaussian distribution and the Nakagami-$m$ distribution. Moreover, the generalized gamma distribution arises as the instantaneous {\effch} distribution in Nakagami-$m$ fading channels with receiver non-linearity, where the non-linearity is captured as a power parameter, since $X$ can be written as \cite[p. 13]{book:bondessonGGC}
\begin{align} 
\label{eqn:GenGammaAndGamma}
X=Y^{1/r}\;, 
\end{align}
where $Y$ is Gamma distributed with shape parameter $\epsilon/r$ and rate $c_{2}$, denoted as $Y \sim Gamma(\epsilon/r,c_{2})$. 

The density of $X$ is given by 
\begin{align}
\label{eqn:GenGammDist}
\pdf{X}{x} = c_{1} x^{\epsilon -1} \exp\left( -c_{2} x^{r}\right)\;,
\end{align}  
where $0 < |r| \leq 1$, $c_{1},c_{2} \geq 0$, and $\epsilon /r >0$. 

It is now shown that this distribution belongs to $\HCM$ and consequently $\GGC$ and $\ID$. Recall that the Gamma distribution belongs to $\HCM$, as shown in Section \ref{sec:GammaDistGGC}. Using Lemma \ref{lem:props_GGC_HCM}, it is then seen that the generalized gamma distribution also belongs to $\HCM$, and is therefore a GGC. While the Thorin measure for the general case is unknown, the diversity order is equal to $\epsilon$ if $r>0$, and $\infty$ otherwise, as can be verified by representing \eqref{eqn:GenGammDist} in the form \eqref{eqn:GGCPDF_canonical}, and further, the function $h(\cdot)$ in \eqref{eqn:GGCPDF_canonical} is given by $h(x) = c_{1}\exp(-c_{2}x^{r})$. 
\subsubsection{Product of Generalized Gamma Random Variables}
The product of independent generalized Gamma RVs has been proposed as a unified fading model in \cite{sagias2006performance}. As a special case of this family of RVs, product of $N$ independent Gamma RVs is obtained, which is the instantaneous channel power distribution in the generalized-$K$ fading model \cite{shankar2004error}. Another special case of this distribution is the $N*$Nakagami fading model \cite{paper:Nstarnakagami, andersen2002statistical}, where the envelope is  a product of $N$ independent but not necessarily identically distributed Nakagami-$m$ RVs.

Consider $X$ to be a product of $N$ independent but not necessarily identically distributed generalized Gamma distributed RVs. This distribution is HCM as observed from an application of Lemma \ref{lem:props_GGC_HCM} and the fact that the generalized Gamma distribution belongs to $\HCM$. As a consequence, the distribution is also in $\GGC$ and $\ID$.
%
%
\subsubsection{Spherically Invariant Random Process}
The spherically invariant random process (SIRP) model has been proposed as a unified model for the fading envelope distribution in the literature \cite{yao2004unified}. In this case, the instantaneous {\effch} $X$ is given by 
\begin{align}
X = A E\;,
\end{align}
where $E$ is an exponential RV, and $A$ is a positive valued RV. 

It is now shown that $X$ for this case is always an ID RV, and is GGC and HCM for certain special cases. To see that $X$ is ID, observe that $X$ is a variance mixture of exponentials, which is ID, according to \cite[Theorem 2.4.3]{book:bondessonGGC}. Not every member of SIRP is a GGC. A simple counterexample is the case of Rician fading, which is a member of SIRP according to \cite{yao2004unified}, and not a GGC as seen in Section \ref{sec:GGCRicianexample}. The special cases of SIRP which are GGC RVs and HCM RVs are now considered. $X$ is a GGC, whenever $A$ is a GGC, according to Lemma \ref{lem:props_GGC_HCM}. Also, $X$ is a HCM RV for all such members of SIRP where $A$ is a HCM RV. 
\subsubsection{Positive Stable Distribution}
Positive stable distributions have been used to model the interference at the primary receiver in a cognitive radio network, when the interfering secondary terminals are distributed in a Poisson field, and there is a guard zone around the primary receiver \cite{hong2008interference}.

The positively skewed stable distribution is a heavy tailed distribution, which is characterized by its LT as $\MGF{X}{s} = \exp(-s^{r})$, $0<r \leq 1$. No closed form expression for the distribution or density is known in general. Nevertheless, it is known that this distribution is a GGC \cite[p. 35]{book:bondessonGGC}, and the density of the Thorin measure is given by 
\begin{align}
\label{eqn:posStableThorinMeas}
\frac{\D\mu(u)}{\D u} =\frac{r \sin r \pi}{\pi}u^{r-1}\;,
\end{align}
and diversity order equal to $\infty$. 

\subsubsection{Pareto Distribution}
Pareto RVs are heavy-tailed distributions used to model signal to interference ratios in interference dominated scenarios. A Pareto RV can be written in the form \cite[p. 14]{book:bondessonGGC}
\begin{align}
\label{eqn:paretoIsGammaRatio}
X = \left(\frac{X_{1}}{X_{2}} \right)^{\frac{1}{r}}\;,
\end{align}
where $X_{j} \sim Gamma(k_{j},1)$, $j=1,2$ are independent.
The fact that Pareto RVs are ratios of Gamma RVs raised to a non-negative power, as shown in \eqref{eqn:paretoIsGammaRatio} leads to many wireless systems with Pareto instantaneous {\effch}. For example, consider a system with a transmitter, receiver and an interfering terminal. Suppose the channel between the transmitter and the receiver, and that between the interferer and receiver are both Nakagami-$m$ channels (with possibly different parameters), and the receiver is non-linear with a power nonlinearity. Then from \eqref{eqn:paretoIsGammaRatio}, it is clear that the instantaneous {\effch} is a Pareto RV. It is also noted that special cases of the Pareto distribution with $(k_{1}=k_{2}=1, r>0)$ have been used to model the instantaneous signal to interference power in interference dominated networks \cite{pun07}.

The density of a Pareto RV is given by
\begin{align}
\label{eqn:ParetoPDF}
\pdf{X}{x} = \frac{|r|}{B(k_{2},k_{1})}\frac{x^{k_{1}r-1}}{(1+x^{r})^{k_{1}+k_{2}}}\;,
\end{align} 
where $B(a,b) := \int_{0}^{1} t^{a-1}(1-t)^{b-1} \D t$ is the Beta function, $r \in [-1,1]$, $k_j > 0$, $j=1,2$. 

The Pareto distribution belongs to $\HCM$. This is because Gamma RVs are HCM, and using Lemma \ref{lem:props_GGC_HCM} on \eqref{eqn:paretoIsGammaRatio}, it follows that a Pareto RV is HCM. The Thorin measure is unknown in the general case, however the diversity order is obtained as $k_{1} r$ \cite[p.~74]{book:bondessonGGC}.

As discussed above, almost all of the fading distributions with ID instantaneous {\effch} also belong to $\GGC$. It is therefore reasonable to focus on channels with instantaneous {\effch} in $\GGC$. Hereafter, we refer to channels with GGC distributed {\effch} as GGC channels.
\section{GGC Fading Channels as Gamma Mixtures}
Let $X$ be a GGC with finite diversity order $D$. As discussed in Lemma \ref{lem:GGC_gammaMixture}, $X$ can be written as $AZ$, where $Z$ is Gamma distributed with parameters $(D,1)$, and $A$ is a nonnegative RV independent of $Z$. Using the property $tZ\sim Gamma(D,1/t)$ applicable to Gamma RVs, it is easy to see that $X$ can be written as $\tilde{A}\tilde{Z}$, where $\tilde{Z}$ is Gamma distributed with parameters $(D,D)$, and $\tilde{A}$ is a nonnegative RV independent of $\tilde{Z}$. This representation shows that a GGC with channel with finite diversity order is a Nakagami-$m$ fading channel with shadowing, where $m=D$. This is because the instantaneous {\effch} of a Nakagami-$m$ channel with $m=D$ has the same distribution as $\tilde{Z}$. Moreover, the average {\effch} of the shadowing RV is $1$, since it has been assumed that $\E{}{X}=1$.

The distribution of the shadowing RV $\tilde{A}$ is not trivial to obtain, when only the distribution of $X$ is known. In this case, we adapt the approach used to obtain the mixing distribution for mixture of exponential distributions \cite{yao2004unified}. This method is based on the fact that the Mellin transform of a product of independent RVs is the product of the Mellin transforms of the RVs, and that the distribution and its Mellin transform pair are unique. The Mellin transform of a nonnegative RV $\tilde{A}$ is defined as
\begin{align}
M_{\tilde{A}}(s) = \int\limits_{0}^{\infty}f_{\tilde{A}}(x)x^{s-1}\D s\;,
\end{align}
and $M_{\tilde{AZ}}(s) = M_{\tilde{A}}(s)M_{\tilde{Z}}(s)$. Therefore, the PDF of $\tilde{A}$ can be obtained using the inverse Mellin transform as follows:
\begin{align}
\label{eqn:GGC_mellinTFinv}
f_{\tilde{A}}(x) = M^{-1}\left[ \frac{M_{X}(s)}{M_{\tilde{Z}}(s)}\right]\;.
\end{align}
The Mellin transform of $\tilde{Z}$ can be evaluated as $m^{1-s}\Gamma(s+D-1)/\Gamma(D)$ \cite[p. 312]{bateman1954tables}. Thus, substituting the Mellin transform of $X$ in \eqref{eqn:GGC_mellinTFinv}, the density of $A$ can be obtained. Reference \cite{yao2004unified} uses this approach to obtain $f_{\tilde{A}}(\cdot)$ in terms of Meijer-G functions .
\section{Performance Metrics for GGC Channels}
\label{sec:GGC_performanceMetrics}
In \cite{paper:yuanTauberian}, the diversity and the asymptotic error rate of a wide range of channel distributions are analyzed using the theory of regular variation.  It has been proved that the diversity order is equivalent to the variation exponent of the CDF of the channel power gain at the origin. The asymptotic average error rate characterization applies to a general set of channel distributions. In this Section, with the establishment of the infinitely divisible RVs, asymptotic expressions for average SER of GGC channels and/or the function $h(\cdot)$ defined in \eqref{eqn:GGCPDF_canonical}, and the diversity orders in terms of the Thorin measure  are obtained.
\subsection{Asymptotic Symbol Error Rate}
\label{sec:GGCasympSER}
The average SER $\AvgPe{}{\rho}$ is defined as follows:
\begin{align}
\label{eqn:GGCAvgPeandPeeasy}
\AvgPe{}{\rho} = \E{}{P_{\mathrm{e}}(\rho X)} \;,
\end{align}
where $P_{\mathrm{e}}(\rho X)$ is the instantaneous SER, and is dependent on the constellation of choice.

The analysis of the asymptotic SER of M-PSK and DPSK are considered. In the asymptotically high SNR regime, $\AvgPe{}{\rho}$ is given by 
\begin{align}
\label{eqn:GGCasympSERexp}
\AvgPe{}{\rho} \sim c_{Q} h(\rho^{-1}) \rho^{-D}\;, \rho \rightarrow \infty\;,
\end{align}
where $h(x)$ is slowly varying at $0$ \cite{paper:yuanTauberian}. Using \eqref{eqn:GGCasympSERexp}, generalized definitions quantities related to asymptotic SER are obtained. The diversity order $D$ is the asymptotic slope of the average SER on a log-log plot versus the average SNR, and is defined as the variation exponent of \eqref{eqn:GGCAvgPeandPeeasy} at $\infty$ \cite{paper:yuanTauberian}.
The term $c_{Q}$ in \eqref{eqn:GGCasympSERexp} is a modulation dependent constant, and is given by $c_{Q} = \pi/4$ for the case of DPSK \cite{paper:yuanTauberian}, and 
\begin{align}
c_{Q}= (\pi \sin^{2D}(M^{-1}\pi))^{-1}\int_{0}^{(1-M^{-1})\pi}\sin^{2D}\theta \D \theta\;,
\end{align}
for the case of MPSK \cite{paper:yuanTauberian}. 

Now the average SER performance of a GGC channel at high SNR is now considered, in order to obtain the diversity order in terms of the Thorin measure. 
\begin{proposition}
\label{thm:GGCSER_Linear}
Let a GGC fading channel have diversity order $D < \infty$. Then for DPSK/MPSK, we have $\AvgPe{}{\rho} \sim  c_{Q} h(\rho^{-1}) \rho^{-D}$, as $\rho \rightarrow \infty$, where $c_{Q}$ depends on the modulation scheme, and $h(\cdot)$ is as defined in \eqref{eqn:GGCPDF_canonical} or \eqref{eqn:GGCGammahxandA}.
\end{proposition} 
\begin{IEEEproof}
According to \cite[Theorem 4]{paper:yuanTauberian}, $\AvgPe{}{\rho}$ can be written as $c_{Q} h(\rho^{-1}) \rho^{-D}$, as $\rho \rightarrow \infty$ for DPSK and MPSK, if the asymptotic density of $X$ has the form $\pdf{X}{x} = x^{D-1}h(x)$, as $x \rightarrow 0$. GGCs with finite diversity order satisfy this condition as seen from Lemma \ref{prop:GGC_density}.
\end{IEEEproof}
It is observed from Proposition \ref{thm:GGCSER_Linear} that the diversity order of a GGC channel is as follows.
\begin{corollary}
The diversity order of DPSK/MPSK over a GGC fading channel with diversity order $D < \infty$ is equal to $D$. 
\end{corollary}
The asymptotic SER obtained in Proposition \ref{thm:GGCSER_Linear} applies to any GGC fading channel with finite diversity order, and shows the dependence of the parameters of the GGC on the asymptotic SER parameters. For GGC channels satisfying a certain condition on the Thorin measure, the asymptotic SER can be written in the form $\AvgPe{}{\rho} \sim {\cg{}\rho}^{-D}$, where $\cg{}$ is known as the array gain, which  represents the shift of the average SER curve to the right on a log-log plot
\begin{proposition}
\label{prop:GGCasympSERalt}
Let the Thorin measure of a GGC fading channel with finite diversity order $D$ satisfy $\int_{1}^{\infty} (\log u) \mu(\D u) < \infty$. Then $\AvgPe{}{\rho} \sim (\cg{}\rho)^{-D}$, as $\rho \rightarrow \infty$, where 
\begin{align}
\label{eqn:GGCarraygain}
\cg{}= \left[ \frac{c_{Q}}{\Gamma(D)}\exp\left( \int\limits_{0}^{\infty} (\log u) \mu(\D u) \right) \right]^{-\frac{1}{D}}\;.
\end{align}
\end{proposition}
\begin{proof}
From Proposition \ref{thm:GGCSER_Linear},
\begin{align}
\label{eqn:GGC_asympSERalteqn1}
\lim_{\rho \rightarrow \infty} \frac{c_{Q}^{-1}\rho^{D}\AvgPe{}{\rho}}{h(\rho^{-1})} =1\;.
\end{align}
The limit of the ratio in \eqref{eqn:GGC_asympSERalteqn1} can be written as a ratio of limits of the numerator and denominator, if and only if $\int_{1}^{\infty} (\log u) \mu(\D u) < \infty$, because this is the condition under which $\lim_{\rho \rightarrow \infty}h(\rho^{-1})< \infty$, as can be seen through an application of Lemma \ref{prop:GGC_density}. Therefore, \eqref{eqn:GGC_asympSERalteqn1} simplifies to
\begin{align}
\label{eqn:GGC_asympSERalteqn2}
\lim_{\rho \rightarrow \infty} c_{\mathrm{M}}^{-1}\rho^{D}\AvgPe{}{\rho} = \lim_{\rho \rightarrow \infty} h(\rho^{-1})\;.
\end{align}
The right hand side of \eqref{eqn:GGC_asympSERalteqn2} is given by $\Gamma(D)^{-1}\exp(\int_{0}^{\infty} (\log u) \mu(\D u))$, as seen from Lemma \ref{prop:GGC_density}. Substituting this in \eqref{eqn:GGC_asympSERalteqn2} and rearranging,
\begin{align}
\label{eqn:GGC_asympSERalteqn3}
\lim_{\rho \rightarrow \infty} \frac{\AvgPe{}{\rho}}{\rho^{-D}c_{\mathrm{M}}\Gamma(D)^{-1}\exp(\int_{0}^{\infty} (\log u) \mu(\D u)) } =1\;.
\end{align}
This proves the Proposition.
\end{proof}
According to Proposition \ref{prop:GGCasympSERalt}, the slow varying function $h(x)$ in Proposition \ref{thm:GGCSER_Linear} becomes a constant at $0$. This conclusion is not necessarily true for fading distributions which are not GGCs, or are GGCs which do not satisfy $\int_{1}^{\infty} (\log u) \mu(\D u) < \infty$. An example of a fading distribution for which the slow varying function $h$ does not become a constant at $0$ is the case of generalized $K$-fading, as discussed in \cite{paper:yuanTauberian}. 

It can be seen from \eqref{eqn:GGCarraygain} that the {\codinggain} is inversely proportional to $h(0+)$, where $h(\cdot)$ is as defined in \eqref{eqn:GGCPDF_canonical}, because the term $\Gamma(D)^{-1} \exp(\int_{0}^{\infty} (\log u) \mu(\D u))$ equals $h(0+)$, according to Lemma \ref{prop:GGC_density}. The {\codinggain} will be useful in calculating the performance difference between two GGC fading channels at high SNR, as discussed in Section \ref{sec:GGCperformanceGap}.


\subsection{Quantifying Asymptotic SER Performance Gaps}
\label{sec:GGCperformanceGap}
In the absence of a unified fading model, it is difficult to obtain the SNR gain of one fading distribution with respect to another, without explicitly obtaining the asymptotic SER expressions for both channels. In what follows, the properties of GGCs make it possible to produce a closed form expression for the SNR gain of one GGC with respect to another, when comparing the SERs at high SNR. While stochastic ordering techniques, which will be illustrated in Section \ref{sec:OrderingGGCs}, are capable of yielding GGC channel comparisons based on very general performance metrics without the need for closed form expressions at every value of average SNR, an inherent limitation of this method is that it is not possible to quantify the exact performance gap of one channel with respect to the other.  

The SNR gain of one channel with respect to another is defined as the difference between the SNR (in dB) for the asymptotic SER of the two channels to be equal, and is denoted as $\snrgap{X_{1},X_{2}}$. Mathematically, 
\begin{align}
\snrgap{X_{1},X_{2}} = \lim_{\rho \rightarrow \infty} 10 \log_{10} \frac{ \AvgPe{(1)}{\rho}}{ \AvgPe{(2)}{\rho}}\;,
\end{align}
where $\AvgPe{(1)}{\cdot}$ and $\AvgPe{(2)}{\cdot}$ are the SERs of the channels with instantaneous {\effch} $X_{1}$ and $X_{2}$ respectively.
\begin{proposition}
\label{thm:GGC_SNRgain}
Let $X_{1}$ and $X_{2}$ correspond to the instantaneous {\effch}s of two GGC channels satisfying $\int_{1}^{\infty}\log u \mu_{j}(\D u) <  \infty$, and $\int_{0}^{\infty}\mu_{j}(\D u) = D$, $j=1,2$. If MPSK or MQAM is employed, the SNR gain is given by
\begin{align}
\label{eqn:snrGapmu}
\snrgap{X_{1},X_{2}} = \frac{10}{D} \left(\log_{10}e \right) \int\limits_{0}^{\infty} (\log u)(\mu_{1}(\D u) - \mu_{2}(\D u))\;,
\end{align}
where $\mu_{j}(\cdot)$ is the Thorin measure of $X_{j}$, $j=1,2$. Equivalently, 
\begin{align}
\label{eqn:snrGaphFunction}
\snrgap{X_{1},X_{2}} = \frac{10}{\DivOrd{}} \log_{10} \frac{h_{1}(0+)}{h_{2}(0+)}\;,
\end{align}
where $h_{j}(\cdot)$ are obtained from the canonical GGC PDF \eqref{eqn:GGCPDF_canonical} of $X_{j}, j=1,2$.
\end{proposition}
\begin{IEEEproof}
Relation \eqref{eqn:snrGaphFunction} is proved first. Then \eqref{eqn:snrGapmu} is obtained using the relation \eqref{eqn:hofzero}. Let $\rho_{1}$ and $\rho_{2}$ denote the average SNRs for channels $X_{1}$ and $X_{2}$ required to obtain the same asymptotic SER. In other words, $(\cg{1}\rho_{1})^{- D} + O(\rho_{1}^{-D-1}) =   (\cg{2}\rho_{2})^{- D} + O(\rho_{2}^{-D-1})$. Relation \eqref{eqn:snrGaphFunction} then follows by ignoring the higher order terms, taking logarithms both sides, and substituting the {\codinggain} expressions from \eqref{eqn:GGCarraygain}. 
\end{IEEEproof}

\section{Stochastic Ordering of GGC Distributions}
\label{sec:OrderingGGCs}
\subsection{Laplace Transform Ordering of GGC}
The performance of wireless systems are quantified by averaging a metric (e.g. bit or symbol error rates, or channel capacity) over the distribution of the random channel. Quantifying the system performance relies on single parameter comparisons between channels using characteristics such as diversity order \cite{simon_alouini05}. This approach does not provide a unified frame work to compare channels across many different performance metrics. The theory of stochastic orders provides a comprehensive framework to compare two RVs. It allows comparing systems using different metrics with properties such as monotonicity, convexity, and complete monotonicity, which shed light into the connections between performance metrics such as error rates and ergodic capacity.

In this Section, the stochastic ordering of GGC fading channels is considered. This will help in comparing two GGC fading channels based on general performance metrics which are either completely monotone (such as SERs), or possess a completely monotone derivative (such as the ergodic capacity). Toward this end, the Laplace transform ordering framework proposed in \cite{paper:adithyaSTpaper} is employed. LT ordering between a pair of instantaneous {\effch} distributions implies that the average SER of a constellation with {\cm} SER, such as the case with MPSK and MQAM, will be ordered at all values of average SNR.  

In the context of LT ordering of GGCs, a duality between the Shannon transform ordering of the Thorin measures with the LT ordering of the fading distributions is straightforward to see:
\begin{proposition}
\label{thm:Lt_Shannon_duality}
Let $X$ and $Y$ be two GGCs with Thorin measures $\mu_{X}$ and $\mu_{Y}$. We have $X \orderl{Lt} Y \iff \mu_{X} \orderl{S} \mu_{Y}$.
\end{proposition}
\begin{proof}
The proof of this proposition follows from \eqref{eqn:LT_expectations} applied to GGCs.
\end{proof}

The connection between the LT order and the Shannon transform order as suggested by Proposition \ref{thm:Lt_Shannon_duality} can be exploited in obtaining new ordering relations between GGC fading distributions, by using the properties of one stochastic order to benefit the other. For instance, by observing that $g(x) = \log(1+\rho x)$ is a Bernstein function for $\rho \geq 0$, it is seen that $\mu_{X} \orderl{Lt} \mu_{Y} \implies \mu_{X} \orderl{S} \mu_{Y}$, for cases when the Thorin masses of $X$ and $Y$ are identical. Now, using Proposition \ref{thm:Lt_Shannon_duality}, it is concluded that $\mu_{X} \orderl{Lt} \mu_{Y} \implies X \orderl{Lt} Y$, if the diversity orders of $X$ and $Y$ are equal. Therefore, the generalized LT ordering of the Thorin measures of two GGCs with equal diversity order implies that the average SER performance at all SNR for the first one is better than that of the second at all SNR. 

Laplace transform ordering of GGCs can also be obtained through the observation that any GGC with finite Thorin mass is a gamma variance mixture. This leads to the following Proposition:
\begin{proposition}
\label{thm:LT_GGC_Mixture}
Let $X = A U_{1}$ and $Y = B U_{2}$ be two GGCs, where $A$ and $B$ are nonnegative RVs, independent of $U_{j} \sim Gamma(m_{j},1)$, $j=1,2$. Then $m_{1} \leq m_{2}$ and $A \orderl{Lt} B$ implies $X \orderl{Lt} Y$.
\end{proposition}
\begin{IEEEproof}
To begin with, observe that the LT of a Gamma RV monotonically decreases with the scale parameter, when the rate is unity. Therefore, if $U_{1}$ and $U_{2}$ are Gamma RVs with shape $m_{1}$ and $m_{2}$ respectively, and unit rate, then $U_{1} \orderl{Lt} U_{2} \iff m_{1} \leq m_{2}$. Let $A$ and $B$ be non-negative RVs independent of $U_{1}$ and $U_{2}$, which satisfy $A \orderl{Lt} B$. Then 
\begin{align}
\label{eqn:GGC_LtProofstep1}
\E{}{\exp(-\rho A U_{1})} &= \E{A}{\E{U_{1}}{\exp(-\rho A U_{1})}}\\
\label{eqn:GGC_LtProofstep2}
& \geq \E{A}{\E{U_{2}}{\exp(-\rho A U_{2})}} \\
\label{eqn:GGC_LtProofstep3}
& \geq \E{B}{\E{U_{2}}{\exp(-\rho A U_{2})}}\;,
\end{align}
where \eqref{eqn:GGC_LtProofstep2} is obtained from \eqref{eqn:GGC_LtProofstep1} by observing that $U_{1} \orderl{Lt} U_{2}$, and $\exp(-x)$ in \eqref{eqn:GGC_LtProofstep1} is a {\cm} function. Similarly, \eqref{eqn:GGC_LtProofstep3} follows from \eqref{eqn:GGC_LtProofstep2}, as $A \orderl{Lt} B$.
\end{IEEEproof}
While Proposition \ref{thm:Lt_Shannon_duality} connects the ordering of the Thorin measures and the LT ordering of the corresponding GGCs, Proposition \ref{thm:LT_GGC_Mixture} enables the LT ordering of pairs of GGC RVs for which the Thorin measure may not be available in closed form. 

\subsection{Comparison of GGC Channels with Equal Diversity Orders}
It is well known that the AWGN (no fading) channel is a benchmark to the performance of any fading channel, with respect to symbol error rates and ergodic capacity. Intuitively, this is because the AWGN channel has infinite diversity order (since a Nakagami-$m$ with $m \rightarrow \infty$ is an AWGN channel). However, when two GGC fading distributions with the same diversity order are to be compared, it is in fact observed that the Nakagami-$m$ fading scenario is the best possible fading channel with respect to symbol error rates of $1$-dimensional or $2$-dimensional constellations (which are convex functions) and the ergodic capacity (which is a concave function), as seen from the following Proposition.
\begin{proposition}
\label{thm:GammaBenchmark}
Let $X$ be a GGC with Thorin mass $D<\infty$ and $E{}{X}=1$. Then $U \orderl{cx} X$, where $U \sim Gamma(D,D)$.
\end{proposition}
\begin{IEEEproof}
According to Lemma \ref{lem:GGC_gammaMixture}, $X$ can be written as $X=A U$, where $U$ is Gamma distributed with shape $D$ and rate $D$, and $A$ is a nonnegative RV independent of $U$. Now, let $g$ be a convex function. Then $\E{U}{g(\E{}{A} U)} \leq \E{}{g(X)}$, as a consequence of Jensen's inequality. $\E{}{A}=1$, because $\E{}{X}=1$, by assumption. The proof is thus concluded.
\end{IEEEproof}
It is now possible to see from Proposition \ref{thm:GammaBenchmark} that among all GGC fading channels with a given diversity order and unit average power, the Nakagami-$m$ fading channel forms the benchmark channel with respect to convex performance metrics such as SERs of $1$-dimensional or $2$-dimensional constellations, and concave metrics such as the ergodic capacity. This is because the instantaneous channel power of a Nakagami-$m$ channel is Gamma distributed with shape $m$ and rate $m$. In other words, if the diversity order of a GGC is $D$, then the average SER of any $1$-dimensional or $2$-dimensional constellation is lower bounded by that of a Nakagami-$m$ channel with $m=D$, and further, the ergodic capacity is upper bounded by that of the Nakagami-$m$ channel with $m=D$.

\section{Systems Involving Multiple GGC RVs}
\label{sec:GGCPerfBoundsSys}
Many systems of practical interest involve combinations of multiple fading components. In what follows, the conditions under which the overall end-to-end instantaneous channel power is a GGC with finite diversity order is obtained. This permits the comparison of two GGC channel systems using Proposition \ref{thm:LT_GGC_Mixture}. This facilitates the use of Proposition \ref{thm:GGC_SNRgain} to quantify the high SNR gain of one GGC channel with respect to another, since this propositions require finite diversity order. 

\subsection{Composite Fading Systems}
Many fading models in wireless communications attempt to capture both the short term fading effect and the long term shadowing effect through a composite fading distribution. In such cases, the instantaneous {\effch} is modelled as a product of two RVs, one corresponding to the short-term effect, and the other corresponding to the long-term fading effect. In other words, the overall instantaneous {\effch} $Z$ can be written as $Z=XY$. For such systems, $Z$ is a GGC if $X$ is a GGC and $Y$ is HCM, as seen through an application of Lemma \ref{lem:props_GGC_HCM}. It can be observed that the diversity order of $Z$ is finite if the diversity order of either $X$ or $Y$ is finite. As a result, common composite models such as the Rayleigh-lognormal model \cite{turkmani1992probability}, the Nakagami-lognormal model \cite{ho1993co}, Weibull-Gamma composite model \cite{bithas2009weibull}, Generalized-$K$ \cite{shankar2004error}, and the Gamma-shadowed generalized Nakagami fading model \cite{alouini1999unified} have GGC instantaneous {\effch}s with finite diversity order. 

It can then be concluded using Proposition \ref{thm:LT_GGC_Mixture} that, for two different composite fading distributions $Z_{1}$ and $Z_{2}$, each with finite diversity order, if the diversity order of $Z_{1}$ is larger than that of $Z_{2}$, then to establish LT ordering of the composite distribution, it is sufficient to establish LT ordering of the mixing distributions. Further, if the Thorin measures of the two composite distributions satisfy $\int_{0}^{\infty} \log u \mu_{j}(\D u) < \infty$, then Proposition \ref{thm:GGC_SNRgain} can be used to quantify the high SNR gain of one system versus the other.
\subsection{Diversity Combining Systems}
\label{sec:GGCpaperDivComb}
Consider a single-input multiple-output diversity combining system with $L$ receive antennas, and complete CSI at the receiver. It is now proved that under different assumptions on the instantaneous channel power of each branch, the end-to-end instantaneous channel power of combining scheme such as MRC. 
In the case of MRC, the instantaneous end-to-end channel power $X_{\mathrm{MRC}} = \sum_{i=1}^{L} X_{i}$ is a GGC whenever $X_{i}$ is a GGC, since a sum of GGC RVs is a GGC. It is therefore straightforward to see that the average SER of DPSK is given by one-half $\MGF{}{\rho}$ in \eqref{eqn:MGF_GGC_gen}, where the Thorin measure is the sum of Thorin measures of each component of the MRC system. The ergodic capacity and average SERs of other $2$-dimensional modulations may however not be tractable in general, for example when the branches are independent but not identically distributed GGCs. In this case, it is possible to obtain a performance lower bound if the diversity order is known. To this end, if every component has finite diversity order, then the diversity order of the MRC system is finite, and is given by $D = \sum_{i=1}^{N} \int \mu_{i}(\D u)$. Now, since $X_{\mathrm{MRC}}$ is a GGC with diversity order $D < \infty$, using Proposition \ref{thm:GammaBenchmark}, it is seen that the performance of the MRC system is lower bounded by that of a Nakagami-$m$ channel with $m=D$, with respect to all convex metrics such as average SERs of $2$-dimensional modulations and concave metrics such as the ergodic capacity.
\section{Simulations}
\label{sec:simulations}
In this section some of the theoretical results are corroborated using Monte-Carlo simulations. In Figure \ref{fig:NakagamiVsPareto2antM125}, the performance of a $2$-branch MRC system where the instantaneous channel powers are Pareto distributed with parameters $(1,1,1.25)$ is simulated, and compared with that of a SISO Nakagami-$m$ channel with $m=2.5$. The average powers of both the systems have been normalized to unity. This simulation demonstrates that the average SER of the MRC system with Pareto distributed branches, which are GGC's with finite diversity order is lower bounded by that of a Nakagami-$m$ channel, as suggested in Section \ref{sec:GGCpaperDivComb}. 

In Figure \ref{fig:GenGammaVeNaka}, the performance of a Nakagami-$m$ channel with $m=2$ is compared with that of a fading channel with generalized gamma distributed instantaneous {\effch} with parameters ($\epsilon=2$, $r=2$, $c_{2}=1$). The parameters of the two distributions have been chosen such that the diversity order is $2$ for both cases. The high SNR gain in dB obtained from the simulation is found to be $\approx 1.7$ dB. This agrees with the theoretically suggested value of $1.505$ dB obtained from \eqref{eqn:snrGaphFunction}, with $h_{1}(0+) = 4$ and $h_{2}(0+) = 2$.

A simulation to provide an intuitive understanding of the structure of the Thorin measure and its effect on the average SER performance has been provided in Figure \ref{GGCthorinmasscomp2massesdiffloc}. To this end, two different GGCs $X = X_{1}+X_{2}$ and $Y = Y_{1}+Y_{2}$, where $X_{1}, X_{2}$ are independent and gamma distributed with parameters $(2,1)$ and $(2,2)$ respectively, and $Y_{1}, Y_{2}$ are independent and gamma distributed with parameters $(1,1/2)$ and $(3,3)$ respectively, are chosen. In this case, the densities of the Thorin measure are as depicted in Figure \ref{fig:GGCthorinmasscomp2massesdifflocMU}. The average SER for $X$ and $Y$ are obtained in figure \ref{GGCthorinmasscomp2massesdiffloc}. It is observed that the average SER of $X$ is consistently less than that of $Y$ at all $\rho \geq 0$. Therefore, it can be inferred that the smaller the support of the density of the Thorin measure, the better the average SER performance is. 
\begin{figure}[htb]
\centering
\includegraphics[width=8.5cm]{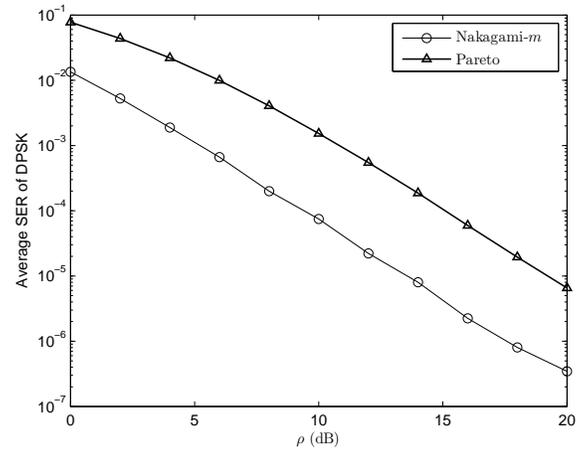}
\caption{Average SER performance of $2$-branch MRC with Pareto($1.25$) distributed SNR compared with that of Nakagami-$m$ SISO channel with $m=2.5$}\label{fig:NakagamiVsPareto2antM125}.
\end{figure}
\begin{figure}[htb]
\centering
\includegraphics[width=8.5cm]{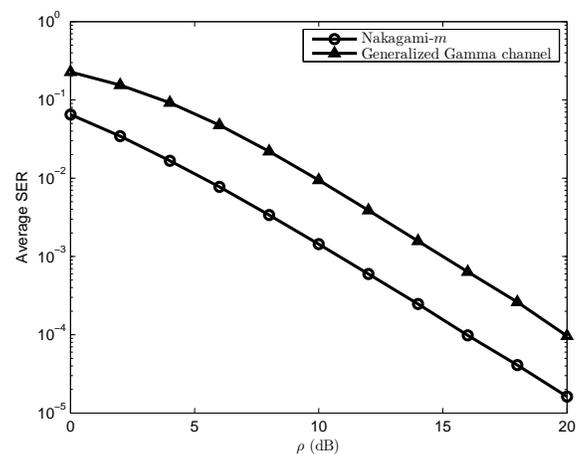}
\caption{Comparison of DPSK performance of Nakagami-$m$ ($m=2$) channel and generalized gamma distributed channel Gamma$(1,1)^{1/2}$ with equal avg. channel powers and diversity orders.}\label{fig:GenGammaVeNaka}.
\end{figure}

\begin{figure}[htb]
\centering
\includegraphics[width=8.5cm]{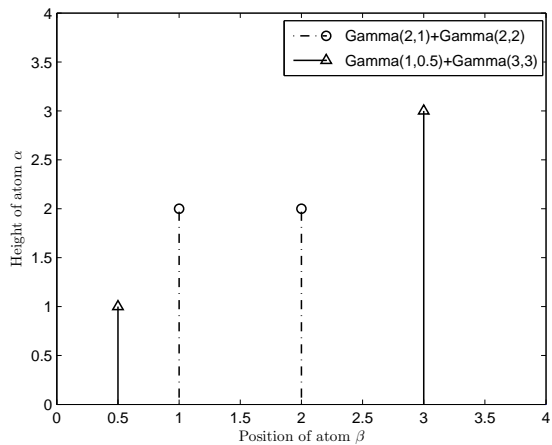}
\caption{Density of Thorin measures for two GGCs. $X = $Gamma$(2,1)+$Gamma$(2,2)$; $Y= $Gamma$(1,0.5)+ $Gamma$(3,3)$. Diversity order = $4$, $\E{}{X}= \E{}{Y}=3$}\label{fig:GGCthorinmasscomp2massesdifflocMU}.
\end{figure}
\begin{figure}[htb]
\centering
\includegraphics[width=8.5cm]{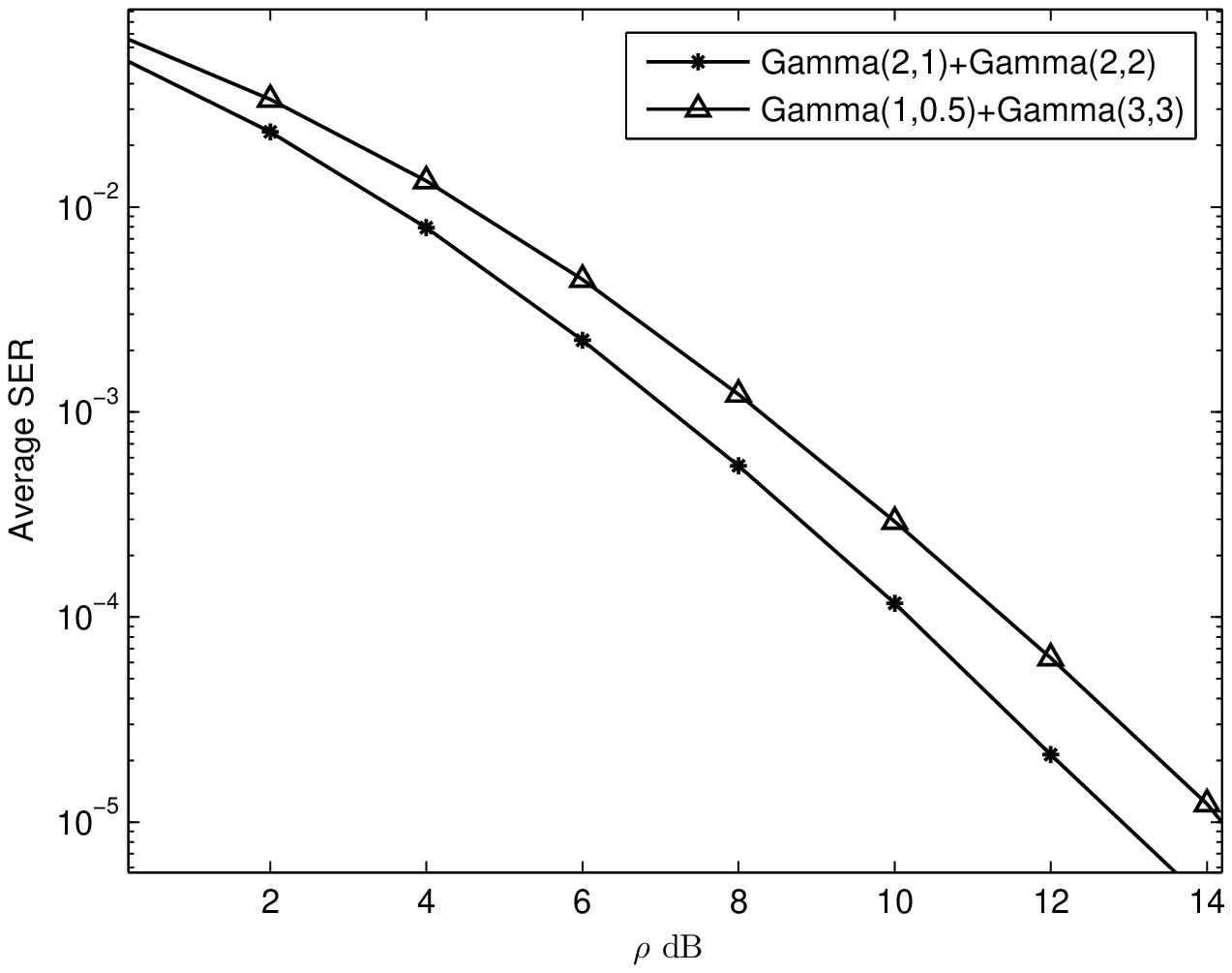}
\caption{SER performance of DPSK over two GGCs which are sums of $2$ gamma RVs with different supports of Thorin measures. The less spread apart $\mu$ is, the better the performance.}\label{GGCthorinmasscomp2massesdiffloc}.
\end{figure}
\section{Conclusions}
\label{sec:conclusions}
In this paper, it has been shown that the class of non-negative ID RVs serves as a unification of commonly used fading distributions such as Rician-$K$, Nakagami-$m$ and Weibull distribution. Furthermore, several unified fading distributions proposed in the literature, such as the Stacy distribution and the SIRP are shown to be included in the class of ID RVs. In addition, it is shown that several ID fading distributions are also GGCs. The properties of ID and GGC RVs find applications in the performance analysis and comparisons of wireless systems. Specifically, the asymptotic average SER is obtained in a canonical form, and its depenence on the Thorin measure of the GGC is revealed. On the other hand, comparing two GGC channels in the LT order is shown to be simplified, by revealing the connection between the Shannon transform ordering of the Thorin measures and the LT ordering of the GGCs. This facilitates comparisons of the average SER or ergodic capacity of the two GGC channels. The properties of GGC RVs also permit us to quantify the asymptotic SER gap for any modulation scheme. Applications of properties of GGCs in diversity combining systems are also provided.

\bibliographystyle{IEEEtran}
\nocite{*}
\bibliography{references}
\vspace{-6mm}
\end{document}